\documentclass[a4paper]{scrartcl}

\pdfoutput=1
\usepackage[english]{babel}
\usepackage{authblk}
\usepackage{booktabs}

\usepackage{amsthm}
\usepackage{amssymb}
\usepackage{amsmath}
\usepackage{xspace}
\usepackage{graphicx}
\usepackage[utf8]{inputenc}
\usepackage{url}
\usepackage[T1]{fontenc}
\usepackage{calculator}
\usepackage[mathscr]{euscript}
\allowdisplaybreaks
\usepackage{breqn}
\newtheorem{definition}{Definition}[section] 
\newtheorem{theorem}{Theorem}[section] 

\DeclareOldFontCommand{\bf}{\normalfont\bfseries}{\mathbf}
\DeclareOldFontCommand{\tt}{\normalfont}{\mathtt}

\newcommand{\rrq}[2]{\DIVIDE{#1}{#2}{\ratio}\ROUND[0]{\ratio}{\ratio}\ratio}

\newcommand{\latte}{LaTTe\xspace}
\newcommand{\azove}{Azove\xspace}

\newcommand{\Nat}{\mbox{${\mathbb N}$}} %% {\mbox{${\rm I\!N}$}}
\newcommand{\Zat}{\mbox{${\mathbb Z}$}} %% {\mbox{${\rm I\!N}$}}
\newcommand{\Rat}
{\mbox{${\mathchoice {\setbox0=\hbox{$\displaystyle\rm Q$}\hbox{\raise
0.15\ht0\hbox to0pt{\kern0.4\wd0\vrule height0.8\ht0\hss}\box0}}
{\setbox0=\hbox{$\textstyle\rm Q$}\hbox{\raise
0.15\ht0\hbox to0pt{\kern0.4\wd0\vrule height0.8\ht0\hss}\box0}}
{\setbox0=\hbox{$\scriptstyle\rm Q$}\hbox{\raise
0.15\ht0\hbox to0pt{\kern0.4\wd0\vrule height0.7\ht0\hss}\box0}}
{\setbox0=\hbox{$\scriptscriptstyle\rm Q$}\hbox{\raise
0.15\ht0\hbox to0pt{\kern0.4\wd0\vrule height0.7\ht0\hss}\box0}}}$}}

\newcommand{\vect}[1]{\ensuremath{\mathbf{#1}}}

\newcommand{\Pre}{\vect{Pre}}
\newcommand{\Post}{\vect{Post}}

\newcommand{\see}[1]{ }

\newcommand{\dotr}[1]{#1^{\bullet}}
\newcommand{\dotl}[1]{^{\bullet}#1}

\newcommand{\fire}[3]{#1 \stackrel{#2}{\rightarrow} #3 }

\usepackage{color}

\definecolor{darkred}{rgb}{0.5,0,0}
\definecolor{darkgreen}{rgb}{0,0.5,0}
\definecolor{darkblue}{rgb}{0,0,0.5}

\definecolor{grey}{rgb}{0.5,0.5,0.5}

\newcommand{\nos}[2]{\ensuremath{(\!({#1})\!)({#2})}}

\newcommand*{\eod}{\hfill\ensuremath{\blacksquare}}
\usepackage{tikz}
\usetikzlibrary{arrows}
\usetikzlibrary{decorations.markings}
\usetikzlibrary{shadows}

\tikzset{
  big stealth/.style={
    decoration={markings,mark=at position -(0.1pt) with {\arrow[scale=2*\scale]{stealth}}},
    postaction={decorate},
    shorten >=0.4pt}}
\tikzset{
  big ring/.style={
    decoration={markings,mark=at position -(0.1pt) with {\arrow[scale=1.5*\scale]{o}}},
    postaction={decorate},
    shorten >=8pt*\scale}}
\tikzset{
  big disc/.style={
    decoration={markings,mark=at position -(0.1pt) with {\arrow[scale=1.5*\scale]{*}}},
    postaction={decorate},
    shorten >=8pt*\scale}}
\tikzset{
  big box/.style={
    decoration={markings,mark=at position -(0.1pt) with {\arrow[scale=1.5*\scale]{open square}}},
    postaction={decorate},
    shorten >=8pt*\scale}}
\tikzset{
  big tile/.style={
    decoration={markings,mark=at position -(0.1pt) with {\arrow[scale=1.5*\scale]{square}}},
    postaction={decorate},
    shorten >=8pt*\scale}}
\tikzstyle{place}=[circle, very thick, fill, top color=white, bottom color=white, draw=black, minimum size=40pt, drop shadow]
\tikzstyle{trans}=[rectangle, very thick, fill, top color=white, bottom color=white, draw=black, minimum size=32pt, drop shadow]
\tikzstyle{arc}=[thick, big stealth, black]
\tikzstyle{read}=[thick, big disc, black]
\tikzstyle{inhibitor}=[thick, big ring, black]
\tikzstyle{stopwatch}=[thick, big tile, black]
\tikzstyle{stopwatchinhibitor}=[thick, big box, black]
\tikzstyle{priority}=[thick, big stealth, orange]
\tikzstyle{enabling}=[thick, big disc, orange]
\tikzstyle{disabling}=[thick, big ring, orange]
\tikzstyle{token}=[circle, fill, draw=black, minimum size=4pt]
\tikzstyle{glob-options}=[label distance=6pt*\scalenodes*\scale,x=1pt,y=-1pt,scale=\scale,every node/.style={transform shape}]
\tikzstyle{virtual}=[circle, draw=white, minimum size=1pt]

\usepackage{hyperref}
\hypersetup{
    colorlinks = true,
    allcolors = {magenta}
}

\newcommand{\reach}{\mathcal{R}}
\newcommand{\netabs}[3]{#1 \mathop{\rhd{\kern -0.5ex}_{#2}} #3}

\newcommand{\vars}{\mathcal{V}}

\newcommand{\V}{\vect{V}}
\newcommand{\uars}{\mathcal{U}}

\newcommand{\lift}{\uparrow}
\newcommand{\proj}{\downarrow}
\newcommand{\qsol}[1]{\langle #1 \rangle}

\begin{document}

%%%%%%%%%%%%%%%%%%%%%%%%%%%%%%%%%%%%%%%%%%%%%%%%%%%%%%%%%%%%%%%%%%%%%%
%%%%%%%%%%%%%%%%%%%%%%%%%%%%%%%%%%%%%%%%%%%%%%%%%%%%%%%%%%%%%%%%%%%%%%

\title{Petri Net Reductions\\ for Counting Markings}

\author[1]{Bernard~Berthomieu\thanks{This material is based upon work supported by the
    RTRA STAE project IFSE2: ``technical engineering for embedded
    systems''.}}  
\author[1]{Didier Le Botlan} 
\author[1]{Silvano~{Dal~Zilio}}
\affil[1]{Université de Toulouse, CNRS, INSA, Toulouse, France}
\date{}
\maketitle
\begin{abstract}
  We propose a method to count the number of reachable markings of a
  Petri net without having to enumerate these first. The method relies
  on a structural reduction system that reduces the number of places
  and transitions of the net in such a way that we can faithfully
  compute the number of reachable markings of the original net from
  the reduced net and the reduction history.
  The method has been implemented and computing experiments show that
  reductions are effective on a large benchmark of models.
\end{abstract}

\bibliographystyle{plain}

%%%%%%%%%%%%%%%%%%%%%%%%%%%%%%%%%%%%%%%%%%%%%%%%%%%%%%%%%%%%%%%%
% \section{Introduction}
\section{Introduction}
\label{sec:introduction}

Structural reductions are an important class of optimization
techniques for the analysis of Petri Nets (PN for short).  The idea is
to use a series of reduction rules that decrease the size of a net
while preserving some given behavioral properties. These reductions
are then applied iteratively until an irreducible PN is reached on
which the desired properties are checked directly. This approach,
pioneered by
Berthelot~\cite{berthelot1985checking,berthelot1986transformations},
has been used to reduce the complexity of several problems, such as
checking for boundedness of a net, for liveness analysis, for checking
reachability properties~\cite{jensen2016tapaal} or for LTL model
checking~\cite{esparza2001net}.

In this paper, we enrich the notion of structural reduction by keeping
track of the relation between the markings of an (initial) Petri net,
$N_1$, and its reduced (final) version, $N_2$. We use reductions of
the form $(N_1,Q,N_2)$, where $Q$ is a system of linear equations that
relates the (markings of) places in $N_1$ and $N_2$. The reductions
are tailored so that the state space of $N_1$ (its set of reachable
markings) can be reconstructed from that of $N_2$ and equations $Q$.
In particular, when $N_1$ is totally reduced ($N_2$ is then the empty
net), the state space of $N_1$ corresponds with the set of non-negative
integer solutions to $Q$.
Then $Q$ acts as a symbolic representation for sets of markings, in much the
same way one can use decision diagrams or SAT-based techniques.

In practice, reductions often lead to an irreducible non-empty
residual net. In this case, we can still benefit from an hybrid
representation combining the state space of the residual net
(expressed, for instance, using a decision diagram) and the symbolic
representation provided by linear equations.
This approach can provide a very compact representation of the state
space of a net. Therefore it is suitable for checking
\emph{reachability properties}, that is whether some reachable marking
satisfies a given set of linear constraints.  However, checking
reachability properties could benefit of more aggressive reductions
since it is not generally required there that the full state space is
available (see e.g. \cite{jensen2016tapaal}).  For particular
properties, it could be more efficient to only approximate the set of
reachable markings, and therefore to use more aggressive reduction
rules (see e.g. \cite{jensen2016tapaal}).

At the opposite, we focus on computing a (symbolic) representation of the full state space.
A positive outcome of our choice is that we can derive a method to count
the number of reachable markings of a net without having to enumerate
them first.

Computing the cardinality of the reachability set has several
applications. For instance, it is a straightforward way to assess the
correctness of tools---all tools should obviously find the same
results on the same models. This is the reason why this problem was
chosen as the first category of examination in the recurring
Model-Checking Contest (MCC)~\cite{mcc:2017,KordonGHPJRH17}. We have
implemented our approach in the framework of the TINA
toolbox~\cite{berthomieu2004tool} and used it on the large set of
examples provided by the MCC (see
Sect.~\ref{sec:experimental-results}). Our results are very
encouraging, with numerous instances of models where our performances
are several orders of magnitude better than what is observed with the
best available tools.

 \vspace{5pt}
\noindent\textbf{Outline.} 
We first define the notations used in the paper then describe the reduction
system underlying our approach, in Sect.~\ref{sec:reductions}. After
illustrating the approach on a full example, in Sect.~\ref{sec:impl-example-appl}, we prove in
Sect.~\ref{sec:correctness} that the equations associated with
reductions allow one to reconstruct the state space of the initial
net from that of the reduced one. Section~\ref{sec:counting} discusses
how to count markings from our representation of a state space while
Sect.~\ref{sec:experimental-results} details our experimental
results. We conclude with a discussion on related works and possible
future directions.

%%%%%%%%%%%%%%%%%%%%%%%%%%%%%%%%%%%%%%%%%%%%%%%%%%%%%%%%%%%%%%%%
% \section{Petri Nets}

\section{Petri Nets}

Some familiarity with Petri nets is assumed from the reader. We recall
some basic terminology. Throughout the text, comparison ($=$, $\ge$)
and arithmetic operations ($-$, $+$) are extended pointwise to
functions.

A marked {\em Petri net} is a tuple
$N = ({P},{T},{\Pre},{\Post}, m_0)$ in which $P$, $T$ are disjoint
finite sets, called the {\em places} and {\em transitions},
$\Pre, \Post : {T} \to ({P} \to \Nat)$ are the {\em pre} and {\em
  post} {\em condition} functions, and $m_0 : P \to \Nat$ is the {\em
  initial marking}.

Figure~\ref{fig:petri} gives an example of Petri net, taken from \cite{Stahl_AWPN},
using a graphical syntax: places are pictured as circles, transitions
as squares, there is an arc from place $p$ to transition $t$ if $\Pre(t)(p) > 0$, and one from
transition $t$ to place $p$ if $\Post(t)(p) > 0$.
The arcs are weighted by the values of the corresponding pre or post conditions (default weight is 1).
The initial marking of the net associates integer 1 to place $p_0$ and 0 to all others.

\begin{figure}[!htb]
  \centering
  \includegraphics[width=0.9\textwidth]{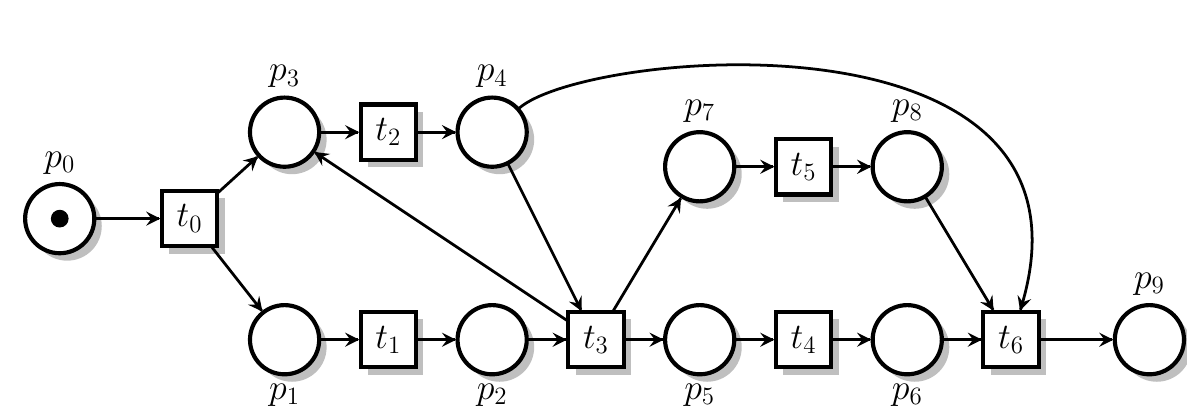}
  \caption{An example Petri net}\label{fig:petri}
\end{figure}

A {\em marking} $m : P \to \Nat$ maps a number of {\em tokens} to every place.
A transition $t$ in $T$ is said {\em enabled} at $m$ if
$m \ge \Pre(t)$.
If enabled at $m$, transition $t$ may {\em fire} yielding a marking
$m'= m - \Pre(t) + \Post(t)$.  This is written
${m \smash[t]{\stackrel{t}{\rightarrow}} m'}$, or simply
$\fire{m}{}{m'}$ when only markings are of interest.
Intuitively, places hold integers and together encode the state (or marking) of a net;
transitions define state changes.

The {\em reachability set}, or {\em state space}, of $N$ is the set of
markings $\reach(N) = \{\,m~\mid~\fire{m_0}{*}{m}\,\}$, where
$\fire{}{*}{}$ is the reflexive and transitive closure of
$\fire{}{}{}$.

A \emph{firing sequence} $\sigma$ over $T$ is a sequence
$t_{1}, \dots, t_{n}$ of transitions in $T$ such that there are some
markings $m_1 \dots, m_{n+1}$ with
$\smash[t]{\fire{m_1}{t_1}{m_2} \land \dots \land
  \fire{m_n}{t_n}{m_{n+1}}}$. This can be written
$\fire{m_1}{\sigma}{m_{n+1}}$.
Its {\em displacement}, or {\em marking change}, is $\Delta(\sigma) = \Sigma_{i=1}^n (\Post(t_i) - \Pre(t_i))$,
where $\Delta : P \to \Zat$, and its {\em hurdle} $H(\sigma)$ is the
smallest marking (pointwise) from which the sequence is firable.\\
Displacements (or marking changes) and hurdles are discussed in \cite{hack1976decidability},
where the existence and uniqueness of hurdles is proved.
As an illustration, the displacement of sequence $t_2t_5t_3$
in the net Figure \ref{fig:petri} is $\{(p_2,-1),(p_5,1),(p_9,1)\}$
(and 0 for all other places, implicitly), its hurdle is $\{(p_2,1),(p_3,1),(p_7,1)\}$.

The {\em postset} of a transition $t$ is
$\dotr{t} = \{\,p \mid \Post(t)(p) > 0\,\}$, its {\em preset} is
$\dotl{t} = \{\,p \mid \Pre(t)(p) > 0\,\}$.  Symmetrically for places,
$\dotr{p} = \{\,t \mid \Pre(t)(p)>0\,\}$ and
$\dotl{p} = \{\,t \mid \Post(t)(p)>0\,\}$.

% \vspace{5pt}

A net is {\em ordinary} if all its arcs have weight one; for all
transition $t$ in $T$, and place $p$ in $P$, we have
$\Pre(t)(p) \le 1$ and $\Post(t)(p) \le 1$.  Otherwise it is said {\em
  generalized}.

A net $N$ is {\em bounded} if there is an (integer) bound $b$ such
that $m(p) \le b$ for all $m \in \reach(N)$ and $p \in P$.  The net is
said {\em safe} when the bound is $1$. All nets considered in this
paper are assumed bounded.

The net in Figure \ref{fig:petri} is ordinary and safe. Its state space
holds 14 markings.

%%%%%%%%%%%%%%%%%%%%%%%%%%%%%%%%%%%%%%%%%%%%%%%%%%%%%%%%%%%%%%%%
% \section{The reduction System}

\section{The Reduction System}\label{sec:reductions}

We describe our set of reduction rules using three main categories.
For each category, we give a property that can be used to recover the
state space of a net, after reduction, from that of the reduced net.

\subsection{Removal of Redundant Transitions}\label{sec:redTrans}

A transition is redundant when its effects can always be achieved by
firing instead an alternative sequence of transitions. Our definition
of redundant transitions slightly strengthens that of
{\em bypass} transitions in \cite{recalde1997improving}.  It is not
fully structural either, but makes it easier to identify special cases
structurally.

\begin{definition}[Redundant transition]\label{def:redTrans}
  Given a net $(P,T,\Pre,\Post,m_0)$, a transition $t$ in $T$ is {\em
    redundant} if there is a firing sequence $\sigma$ over
  $T \setminus \{t\}$ such that $\Delta(t) = \Delta(\sigma)$ and
  $H(t) \ge H(\sigma)$.\eod
\end{definition}

There are special cases that do not require to explore combinations of
transitions. This includes {\em identity} transitions, such that
$\Delta(t) = 0$, and {\em duplicate} transitions, such that for some other
transition $t'$ and integer $k$, $\Delta (t) = k.\Delta(t')$.
Finding redundant transitions using Definition
\ref{def:redTrans} can be convenient too, provided the candidate
$\sigma$ are restricted (e.g. in length). Figure~\ref{fig:redundant}
(left) shows some examples of redundant transitions.  Clearly,
removing a redundant transition from a net does not change its state
space.

\begin{theorem}\label{the:redTrans}
  If net $N'$ is the result of removing some redundant transition in net $N$ then
  $\reach(N) = \reach(N')$
  % $(\forall m)(m \in \reach(N) \Leftrightarrow m \in \reach(N'))$
\end{theorem}

\begin{figure}[!htb]
  \centering
  \begin{tabular}{c@{~~~~}c}
    \begin{tabular}{c} 
      \includegraphics[width=0.45\textwidth]{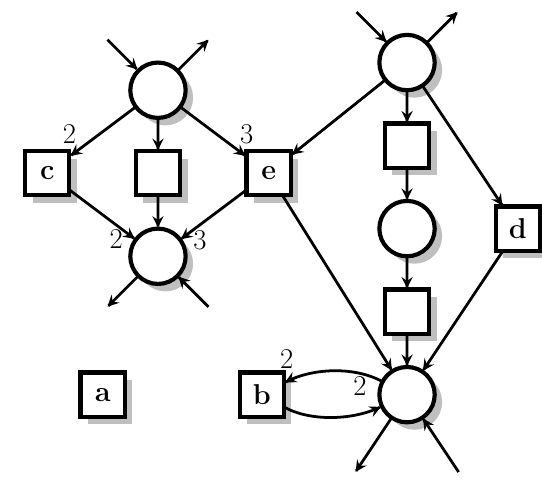}\\
      identity (a,b), duplicate (c)\\
      general redundant (d,e)  \end{tabular} &
    \begin{tabular}{c} 
      \includegraphics[width=0.45\textwidth]{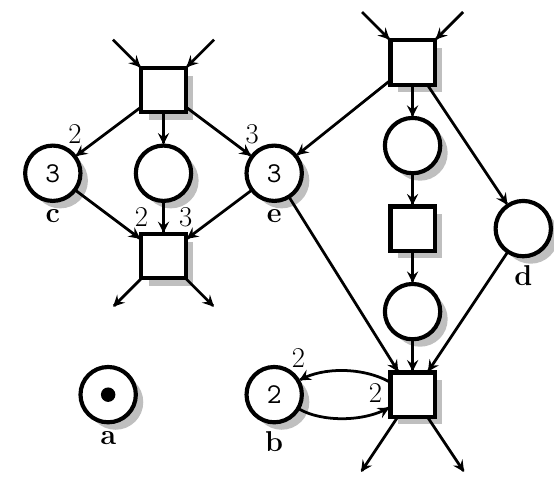}\\
      constant (a,b), duplicate (c)\\
      general redundant (d,e) \end{tabular} \\
  \end{tabular}
  \caption{Some examples of redundant transitions (left) and places (right)}\label{fig:redundant}
\end{figure}

\subsection{Removal of Redundant Places}\label{sec:redPlace}

A place is redundant if it never restricts the firing of its output transitions.
Removing redundant places from a net preserves its language of firing sequences \cite{berthelot1986transformations}.
We wish to avoid enumerating marking for detecting such places, and further be able to recover the marking of
a redundant place from those of the other places.
For these reasons, our definition of redundant places is a slightly strengthened version of that of {\em structurally redundant} places
in \cite{berthelot1986transformations} (last clause is an equation).

\begin{definition}[redundant place]\label{def:redPlace}
  Given a net $(P,T,\Pre,\Post,m_0)$, a place $p$ in $P$ is
  \emph{redundant} if there is some set of places $I$ from
  $P \setminus \{p\}$, some valuation
  $v : (I \cup \{p\}) \rightarrow (\Nat - \{0\})$, and some constant
  $b \in \Nat$ such that, for any $t \in T$:
  \begin{enumerate}
  \item The weighted initial marking of $p$ is not smaller than that
    of $I$:\par $b = v(p).m_0(p) - \Sigma_{q \in I} v(q).m_0(q)$
  \item To fire $t$, the difference between its weighted precondition on $p$ and that on $I$
     may not be larger than $b$:
     $v(p).\Pre(t)(p) - \Sigma_{q \in I} v(q).\Pre(t)(q) \le b$
   \item When $t$ fires, the weighted growth of the marking of $p$ is
     equal to that of $I$:\par
     $v(p).(\Post(t)(p) - \Pre(t)(p)) = \Sigma_{q \in I}
     v(q).(\Post(t)(q) - \Pre(t)(q))$ \eod
  \end{enumerate}
\end{definition}

This definition can be rephrased as an integer linear programming
problem~\cite{silva1996linear}, convenient in practice for computing
redundant places in reasonably sized nets (say when $|P|\ \le 50$).
Like with redundant transitions, there are special
cases that lead to easily identifiable redundant places.  These are
{\em constant} places---those for which set $I$ in the definition is
empty---and {\em duplicated} places, when set $I$ is a
singleton. Figure \ref{fig:redundant} (right) gives some examples of
such places.

From Definition \ref{def:redPlace}, we can show that the marking of a
redundant place $p$ can always be computed from the markings of the
places in $I$ and the valuation function $v$. Indeed, for any marking
$m$ in $\reach(N)$, we have
$v(p).m(p) = \Sigma_{q \in I} v(q).m(q) + b$, where the constant $b$
derives from the initial marking $m_0$. Hence we have a relation
$k_p.m(p) = \rho_p(m)$, where $k_p = v(p)$ and $\rho_p$ is some linear
expression on the places of the net.

\begin{theorem}\label{the:redPlace}
  If $N'$ is the result of removing some redundant place $p$ from net
  $N$, then there is an integer constant $k \in \Nat^*$, and a linear
  expression $\rho$, such that, for all marking $m$:
  $m \cup \{(p,(1/k).\rho(m))\} \in \reach(N) \Leftrightarrow m \in
  \reach(N')$.
\end{theorem}

\subsection{Place Agglomerations}\label{sec:agglo}

Conversely to the rules considered so far, place agglomerations do not
preserve the number of markings of the nets they are applied to. They
constitute the cornerstone of our reduction system; the purpose of the
previous rules is merely to simplify the net so that agglomeration
rules can be applied. We start by introducing a convenient notation.

\begin{definition}[Sum of places]\label{def:sum}
  A place $a$ is the sum of places $p$ and $q$, written
  $a = p \boxplus q$, if: $m_0(a) = m_0(p) + m_0(q)$ and, for all
  transition $t$, $\Pre(t)(a)=\Pre(t)(p)+\Pre(t)(q)$ and
  $\Post(t)(a)=\Post(t)(p)+\Post(t)(q)$.\eod
\end{definition}

Clearly, operation $\boxplus$ is commutative and associative. We
consider two categories of \emph{place agglomeration} rules; each one
consisting in the simplification of a sum of places. Examples are
shown in Fig.~\ref{fig:agglom}.

\begin{figure}[!tb]
  \centering
  \begin{tabular}{c@{~~~~}c@{~~~~}c}
    \begin{tabular}{c} 
      \includegraphics[height=12em]{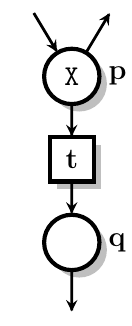} \end{tabular} 
    &
      \begin{tabular}{c} {\Large $\rightarrow$} \end{tabular} 
    &
      \begin{tabular}{c}
        \includegraphics[height=12em]{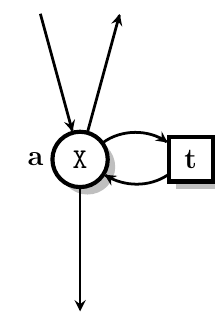} \end{tabular}\\
    \begin{tabular}{c}
      \includegraphics[height=12em]{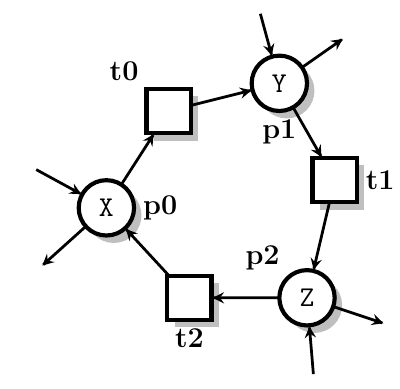} \end{tabular}
    &
      \begin{tabular}{c} {\Large $\rightarrow$} \end{tabular} 
    &
      \begin{tabular}{c} \includegraphics[height=12em]{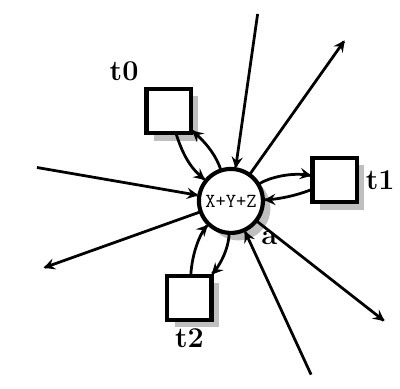}
      \end{tabular}
  \end{tabular}
 \caption{Agglomeration examples: chain (top), loop (for $n=3$, bottom)}\label{fig:agglom}
\end{figure}

\begin{definition}[Chain agglomeration]\label{def:chain}
  Given a net $(P,T,\Pre,\Post,m_0)$, a pair of places $p$, $q$ in $P$
  can be {\em chain agglomerated} if there is some $t \in T$ such
  that: $\dotl{t} = \{p\}$; $\dotr{t} = \{q\}$;
  $\Pre(t)(p) = \Post(t)(q) = 1$; $\dotl{q} = \{t\}$; and
  $m_0(q) = 0$. Their agglomeration consists of replacing places $p$
  and $q$ by a place $a$ equal to their sum: $a = p \boxplus q$.  \eod
\end{definition}

\begin{definition}[Loop agglomeration]\label{def:loop}
  A sequence of $n$ places $(\pi_i)_{i=0}^{n-1}$ can be {\em loop
    agglomerated} if the following condition is met:
    \[
    (\forall i < n)(\exists t \in T)(\Pre(t)=\{(\pi_{i},1)\} \land
    \Post(t)=\{(\pi_{(i + 1)(\text{mod}\ {n})},1)\})~.\]
  Their agglomeration consists of replacing places
  $\pi_0, \dots, \pi_{n-1}$ by a single place, $a$, defined as their
  sum: $a = \boxplus_{i=0}^{n-1} {\pi_i}$.\eod
\end{definition}

Clearly, whenever some place $a$ of a net obeys $a = p \boxplus q$ for
some places $p$ and $q$ of the same net, then place $a$ is redundant
in the sense of definition \ref{def:redPlace}.
The effects of agglomerations on markings are stated by Theorem \ref{the:agglo}.

\begin{theorem}\label{the:agglo}
  Let $N$ and $N'$ be the nets before and after agglomeration of some
  set of places $A$ as place $a$. Then for all markings $m$ over
  $(P \setminus A)$ and $m'$ over $A$ we have:
  $(m \cup m') \in \reach(N) \Leftrightarrow m \cup \{(a,\Sigma_{p \in
    A}m'(p)\} \in \reach(N')$.
\end{theorem}

\begin{proof}  
  Assume $N$ is a net with set of places $P$. Let us first consider
  the case of the chain agglomeration rule in Fig.~\ref{fig:agglom}
  (top). We have to prove that for all marking $m$ of
  $P \setminus \{p, q\}$ and for all values $x,y$ in $\Nat$:
  \[
    m \cup \{(p,x),(q,y)\} \in \reach(N) \Leftrightarrow m \cup
    \{(a,x+y)\} \in \reach(N')
  \]

 \noindent\textbf{Left to right (L):} Let $N^+$ be net $N$ with place $a = p \boxplus q$ added.  Clearly,
 $a$ is redundant in $N^+$, with $v(a)= v(p) = v(q)=1$. So $N$ and
 $N^+$ admit the same firing sequences, and for any
 $m \in \reach(N^+)$, we have $m(a) = m(p) + m(q)$.  Next, removing
 places $p$ and $q$ from $N^+$ (the result is net $N'$) can only relax
 firing constraints, hence any $\sigma$ firable in $N^+$ (and thus in
 $N$) is also firable in $N'$, which implies the goal.\\

 \noindent\textbf{Right to left (R):} we use two intermediate properties ($\forall m,x,u,v$
 implicit).  We write $m \sim m'$ when $m$ and $m'$ agree on all
 places except $p$, $q$ and $a$, and $m \approx m'$ when
 $m \sim m' \land m(p) = m'(a) \land m(q) = 0$.\par

 \vspace{3pt}
 \noindent Property (1):
 $m \cup \{(a,x)\} \in \reach(N') \Rightarrow m \cup \{(p,x),(q,0)\}
 \in \reach(N)$.
    
 \vspace{3pt} Since $\Delta(t)=0$, any marking reachable in $N'$ is
 reachable by a sequence not containing $t$, call these sequences
 $t$-free. Property (1) follows from a simpler relation, namely (Z):
 \emph{whenever $m \approx m'$ ($m \in \reach(N)$, $m'\in \reach(N')$) and
 $\smash[t]{\fire{m'}{\delta}{w'}}$, ($\delta$ $t$-free), then there
 is a sequence $\omega$ such that $\fire{m}{\omega}{w}$ and
 $w \approx w'$}.
  
 Any $t$-free sequence firable in $N'$ but not in $N$ can be written
 $\sigma.t'.\gamma$, where $\sigma$ is firable in $N$ and transition
 $t'$ is not firable in $N$ after $\sigma$. Let $w$, $w'$ be the
 markings reached by $\sigma$ in $N$ and $N'$, respectively.  Since
 $\sigma$ is firable in $N$, we have $w \approx w'$, by (L) and the
 fact that $\sigma$ is $t$-free (only $t$ can put tokens in $q$). That
 $t'$ is not firable at $w$ but firable at $w'$ is only possible if
 $t'$ is some output transition of $a$ since $w \sim w'$ and the
 preconditions of all other transitions of $N'$ than $a$ are identical
 in $N$ and $N'$.  That is, $t'$ must be an output transition of
 either or both $p$ or $q$ in $N$. If $t'$ has no precondition on $q$
 in $N$, then it ought to be firable at $w$ in $N$ since
 $w(p) = w'(a)$.  So $t'$ must have a precondition on $q$; we have
 $w(q) \not\ge \Pre(t')(q)$ in $N$ and $w'(a) \ge \Pre'(t')(a)$ in
 $N'$. Therefore, we can fire transition $t$ $n$ times from $w$ in
 $N$, where $n = \Pre(t')(q)$, since $w'(a) = w(p)$ and $t'$ is
 enabled at $w'$, and this leads to a marking enabling $t'$.  Further,
 firing $t'$ at that marking leaves place $q$ in $N$ empty since only
 transition $t$ may put tokens in $q$. Then the proof of Property (1)
 follows from (Z) and the fact that Definition \ref{def:chain} ensures
 $m_0 \approx m_0'$.

 \vspace{3pt}
 \noindent Property (2): if $m \cup \{(p,x),(q,0)\} \in \reach(N)$ and
 $(u + v = x)$ then $m \cup \{(p,u),(q,v)\} \in \reach(N)$.

 \vspace{3pt} Obvious from Definition \ref{def:chain}: the tokens in
 place $p$ can be moved one by one into place $q$ by firing $t$ in
 sequence $v$ times.

 Combining Property (1) and (2) is enough to prove (R), which
 completes the proof for chain agglomerations. The proof for loop
 agglomerations is similar.%\qed
\end{proof}

\subsection{The Reduction System}\label{sec:system}

The three categories of rules introduced in the previous sections
constitute the core of our reduction system.  Our implementation
actually adds to those a few special purpose rules. We mention three
examples of such rules here, because they
play a significant role in the experimental results of
Sect.~\ref{sec:experimental-results}, but without technical details.
These rules are useful on nets
generated from high level descriptions, that often exhibit translation
artifacts like dead transitions or source places.

The first extra rule is the {\em dead transition removal} rule.  It is
sometimes possible to determine statically that some transitions of a
net are never firable. A fairly general rule for identifying
statically dead transitions is proposed in~\cite{esparza2001net}.
Removal of statically dead transitions from
a net has no effects on its state space.

A second rule allows us to remove a transition $t$ from a net $N$
when $t$ is the sole transition enabled in the initial marking and $t$
is firable only once. Then, instead of counting the markings reachable
from the initial marking of the net, we count those reachable from
the output marking of $t$ in $N$ and add~1. Removing such transitions
often yields structurally simpler nets.

Our last example is an instance of simple rules that can be used to do
away with very basic (sub-)nets, containing only a single place. This
is the case, for instance, of the source-sink nets defined below.
These rules are useful if we want to fully reduce a net.  We
say that a net is \emph{totally reduced} when its set of places and
transitions are empty ($P = T = \emptyset$).

\begin{definition}[Source-sink pair]\label{red:sourcesink}
  A pair $(p,t)$ in net $N$ is a {\em source-sink pair} if $\dotl{p} = \emptyset$, $\dotr{p} = \{t\}$, $\Pre(t) = \{(p,1)\}$ and $\Post(t) = \emptyset$.
    \eod
\end{definition}

\begin{theorem}[Source-sink pairs]\label{the:sourcesink}
  If $N'$ is the result of removing a source-sink pair $(p,t)$ in net $N$ then 
    $(\forall z \le m_0(p))(\forall m)(m \cup \{(p,z)\} \in \reach(N) \Leftrightarrow m \in \reach(N'))$.
\end{theorem}

Omitting for the sake of clarity the first two extra rules mentioned
above, our final reduction system resumes to removal of redundant
transitions (referred to as the T rule) and of redundant places (R
rule), agglomeration of places (A rules) and removal of source-sink
pairs (L rule).

Rules T have no effects on markings. For the other three rules, the
effect on the markings can be captured by an equation or an
inequality.
These have shape $v_p.m(p) = \sum_{q \neq p} v_q.m(q) + b$
for redundant places, where $b$ is a constant,
shape $m(a) = \Sigma_{p \in A}
m(p)$ for agglomerations, and shape $m(p) \le
k$ for source-sink pairs, where $k$ is some constant.
In all these equations, the terms
$m(q)$ are marking variables; variable
$m(q)$ is associated with the marking of place
$q$. For readability, we will often use the name of the place instead
of its associated marking variable.  For instance, the marking
equation $2.m(p) = 3.m(q) +
4$, resulting from a (R) rule, would be simply written $2.p = 3.q +
4$.

We show in Sect.~\ref{sec:correctness} that the state space of a net
can be reconstructed from that of its reduced net and the set of
(in)equalities collected when a rule is applied. Before considering
this result, we illustrate the effects of reductions on a full
example.

%%%%%%%%%%%%%%%%%%%%%%%%%%%%%%%%%%%%%%%%%%%%%%%%%%%%%%%%%%%%%%%%
% \section{An illustrative example - HouseConstruction}\label{sec:example}

\section{An Illustrative Example --- HouseConstruction}
\label{sec:impl-example-appl}

\begin{figure}[!tb]
  \makebox[\textwidth][c]{\includegraphics[width=\textwidth]{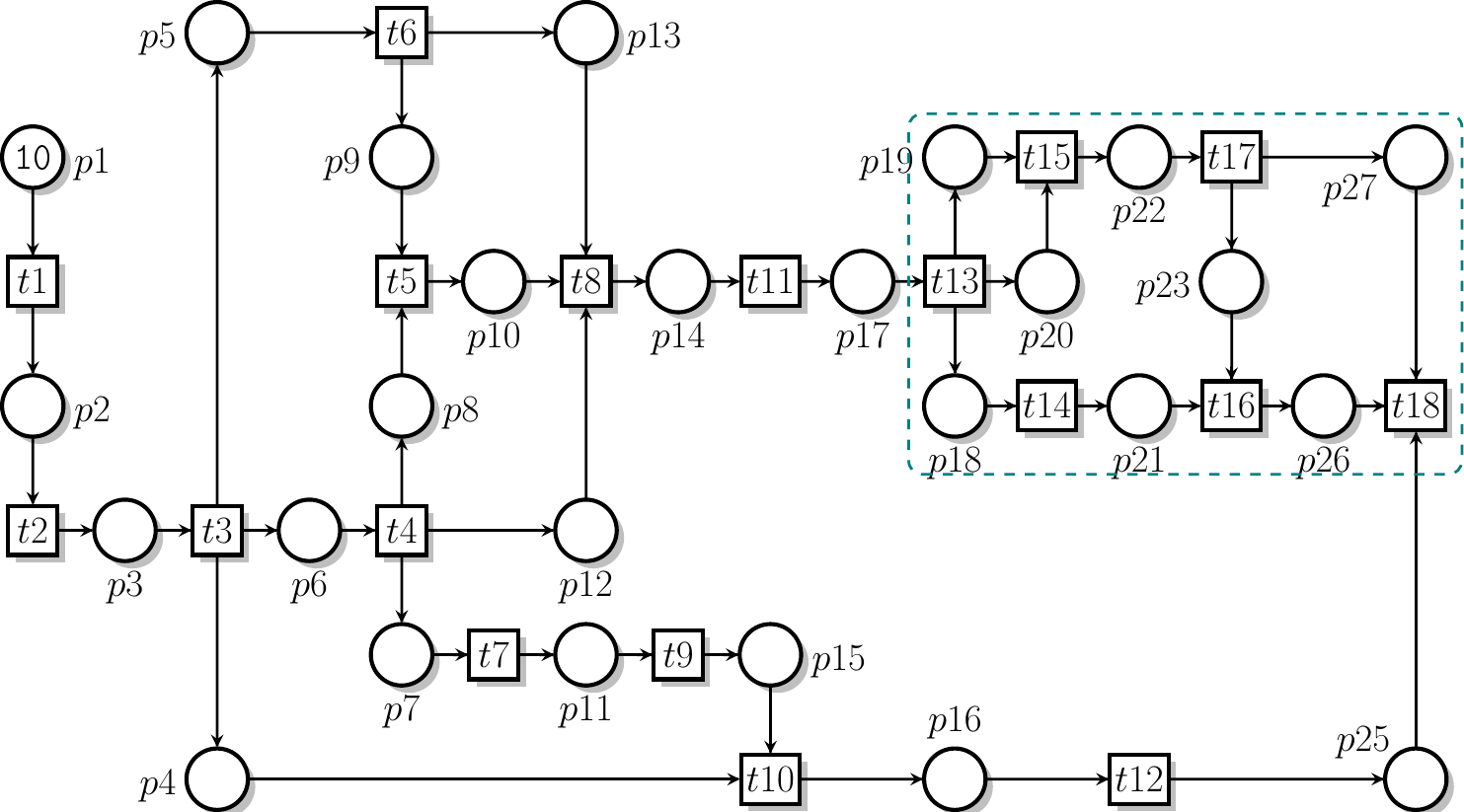}}
  \caption{HouseConstruction-10 example net}.\label{net:HouseConstruction-10}
\end{figure}

We take as example a model provided in the {\em Model Checking
  Contest} (MCC, \url{http://mcc.lip6.fr}), a recurring competition of
model-checking tools~\cite{mcc:2017}.  This model is a variation of a
Petri net model found in \cite{Peterson:1981:PNT:539513}, which is
itself derived from the PERT chart of the construction of a house
found in \cite{levy1963introduction}. The model found in the MCC
collection, reproduced in Fig.~\ref{net:HouseConstruction-10}, differs
from that of \cite{Peterson:1981:PNT:539513} in that it omits time
constraints and a final sink place. In addition, the net represents
the house construction process for a number of houses simultaneously
rather than a single one.  The number of houses being built is
represented by the marking of place $p_1$ of the net ($10$ in the net
represented in Fig.~\ref{net:HouseConstruction-10}).

We list in Fig.~\ref{trace:HouseConstruction-10} a possible reduction
sequence for our example net, one for each line.
To save space, we have omitted the removal of redundant transitions.  For each
reduction, we give an indication of its kind (R, A, \dots), the
marking equation witnessing the reduction, and a short description.
The first reduction, for instance, says that place $p_{19}$ is
removed, being a duplicate of place $p_{20}$. At the second step,
places $p11$ and $p7$ are agglomerated as place $a1$, a ``fresh''
place not found in the net yet.

\begin{figure}[!htb]
\noindent{\tt\small
  \begin{tabular}{l@{~}|@{~}l}
  \begin{tabular}{l@{~~}l}
    R |- p19 = p20          & p19 duplicate\\
    A |- a1 = p11 + p7      & agglomeration\\
    A |- a2 = p17 + p14     & agglomeration\\
    A |- a3 = p2 + p1       & agglomeration\\
    A |- a4 = p21 + p18     & agglomeration\\
    A |- a5 = p22 + p20     & agglomeration\\
    A |- a6 = p25 + p16     & agglomeration\\
    A |- a7 = p15 + a1      & agglomeration\\
    A |- a8 = p3 + a3       & agglomeration\\
    R |- p12 = p10 + p8     & p12 redundant\\
    R |- p13 = p10 + p9     & p13 redundant\\
    R |- a4 = a5 + p23      & a4 redundant\\
    R |- p27 = p23 + p26    & p27 redundant
  \end{tabular} &
  \begin{tabular}{l@{~~}l}
    R |- p4 = p6 + a7       & p4 redundant\\
    A |- a9 = a2 + p10      & agglomeration\\
    A |- a10 = a6 + a7      & agglomeration\\
    A |- a11 = p23 + a5     & agglomeration\\
    A |- a12 = p9 + p5      & agglomeration\\
    %%%%%%%%%%%%%%%%%%%%%%%%%%%%%%%%%%%%%%%%%%%%%%%%%%%%%%%%
    \iffalse
    A |- a13 = a11 + a9     & agglomeration\\
    A |- a14 = p26 + a13    & agglomeration\\    
    \fi
    %%% Variante pour isoler le sous-reseau de la figure
    A |- a13 = a11 + p26    & agglomeration\\
    A |- a14 = a13 + a9     & agglomeration\\    
    %%%%%%%%%%%%%%%%%%%%%%%%%%%%%%%%%%%%%%%%%%%%%%%%%%%%%%%%
    R |- a12 = p6 + p8      & a12 redundant\\
    R |- a10 = a14 + p8     & a10 redundant\\
    A |- a15 = a14 + p8     & agglomeration\\
    A |- a16 = p6 + a8      & agglomeration\\
    A |- a17 = a15 + a16    & agglomeration\\
    L |- a17 <= 10          & a17 source
  \end{tabular}
  \end{tabular}
}
  \caption{Reduction traces for net HouseConstruction-10}\label{trace:HouseConstruction-10}
\end{figure}

Each reduction is associated with an equation or inequality linking
the markings of the net before and after application of a rule.  The
system of inequalities gathered is shown below, with agglomeration
places $a_i$ eliminated. We show in the next section that the set of
solutions of this system, taken as markings, is exactly the set of
reachable markings of the net.

\[\begin{array}{c}
    \begin{array}{rcl@{~~~~~~~~~~~~~~~~}rcl}
      p19 & = & p20                     & p4 & = & p6 + p15 + p11 + p7\\
      p12 & = & p10 + p8                & p9 + p5 & = & p6 + p8 \\
      p13 & = & p10 + p9                & p21 + p18 & = & p22 + p20 + p23 \\
      p27 & = & p23 + p26               & 
    \end{array}\\
    \vspace{-6pt}
    \begin{array}{rcl}
      p25 + p16 + p15 + p11 + p7 & = & p26 + p23 + p22 + p20 + p17 + p14 + p10 + p8\\
      p26 + p23 + p22 + p20 + p17 & + & p14 + p10 + p8 + p6 + p3 + p2 + p1 \le 10\\
    \end{array}\\
  \end{array}\]

\vspace{12pt}
This example is totally reduced using the sequence of
reductions listed. And we have found other examples of totally
reducible net in the MCC benchmarks. In the general case, our
reduction system is not complete; some nets may be only partially
reduced, or not at all.

When a net is only partially reducible, the inequalities, together
with an explicit or logic-based symbolic description of the
reachability set of the residual net, yield a hybrid representation of
the state space of the initial net. Such hybrid representations are
still suitable for model checking reachability properties or counting
markings.

\vspace{5pt}
\noindent\textbf{Order of application of reduction rules.}
Our reduction system does not constrain the order in which reductions are
applied. Our tool attempts to apply them in an order that minimizes reduction costs.

The rules can be classified into ``local'' rules, detecting some structural patterns
on the net and transforming them, like removal of duplicate transitions or places,
or chain agglomerations, and `'non-local'' rules, like removal of redundant places
in the general case (using integer programming).
Our implementation defers the application of the non-local rules until no more local rule
can be applied. This decreases the cost of non-local reductions as they are applied to smaller nets.

Another issue is the confluence of the rules. Our reduction system is not confluent:
different reduction sequences for the same net could yield different residual nets.
This follows from the fact that agglomeration rules do not preserve in  general the
{\em ordinary} character of the net (that all arcs have weight 1), while agglomeration
rules require that the candidate places are connected by arcs of weight 1 to the same transition.

An example net exhibiting the problem is shown in
Fig.~\ref{fig:confluence}(a).  Agglomeration of places $p3$ and $p4$
in this net, followed by removal of identity transitions, yields the
net in Fig.~\ref{fig:confluence}(b).  Place $a1$ in the reduced net is
the result of agglomerating $p3$ and $p4$; this is witnessed by
equation $a1 = p3 + p4$.
Note that the arcs connecting place $a1$ to transitions $t0$ and $t1$ both have weight 2.

\begin{figure}[!htb]
  \centering
  \begin{tabular}{c@{~~~~~~~~~~~~~~~~~}c}
    \begin{tabular}{c} 
      \includegraphics[height=20em]{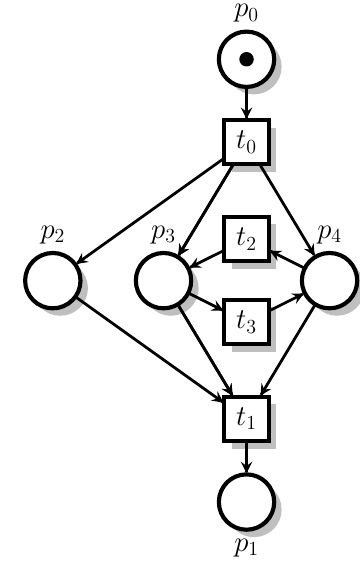}\\
      (a) \end{tabular} &
    \begin{tabular}{c} 
      \includegraphics[height=20em]{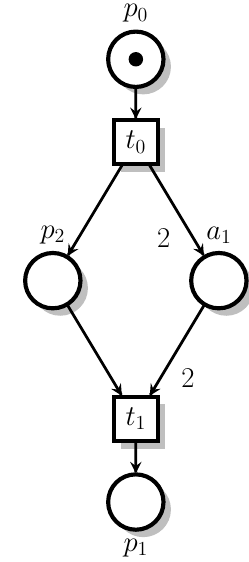}\\
      (b) \end{tabular} \\
  \end{tabular}
  \caption{Non confluence example}\label{fig:confluence}
\end{figure}

Next, place $p2$ in the reduced net is a duplicate of place $a1$,
according to the definitions of Sect.~\ref{sec:redPlace}, the
corresponding equation is $2.p2 = a1$.  But, from the same equation,
$a1$ is a duplicate of $p2$ as well. But removing $p2$ or $a1$ have
differents effects:

\begin{itemize}
\item If $a1$ is removed, then we can fully reduce the net by the
  following sequence of reductions:

\begin{quote}
\noindent{\tt\small
\begin{tabular}{l@{~~}l}
    A |- a2 = p1 + p2  & agglomeration\\
    A |- a3 = a2 + p0  & agglomeration\\
    R |- a3 = 1          & constant place
\end{tabular}
}
\end{quote}

\item If $p2$ is removed instead, then the resulting net cannot be
  reduced further: places $p0$, $a1$ and $p1$ cannot be agglomerated
  because of the presence of arcs with weight larger than 1.
\end{itemize}

Confluence of the system could be easily obtained by restricting the
agglomeration rules so that no arcs with weight larger than 1 could be
produced. But it is more effective to favour the expressiveness of our
reduction rules.

Alternatively, agglomeration rules could be generalized to handle arbitrary weights on the
arcs linking the agglomerated places; this is a scheduled improvement of our reduction system.

%%%%%%%%%%%%%%%%%%%%%%%%%%%%%%%%%%%%%%%%%%%%%%%%%%%%%%%%%%%%%%%%
% \section{Correctness of markings reconstruction}\label{sec:correctness}

\section{Correctness of Markings Reconstruction}\label{sec:correctness}

We prove that we can reconstruct the markings of an (initial) net,
before application of a rule, from that of the reduced net. This
property ensues from the definition of a \emph{net-abstraction}
relation, defined below. 

We start by defining some notations useful in our proofs. We use
$\uars, \vars, \dots$ for finite sets of non-negative integer
variables. We use $Q, Q'$ for systems of linear equations (and
inequalities) and the notation $\V(Q)$ for the set of variables
occurring in $Q$. The system obtained by concatenating the relations
in $Q_1$ and $Q_2$ is denoted $(Q_1; Q_2)$ and the ``empty system'' is
denoted $\emptyset$.

A valuation $e$ of $\Nat^\vars$ is a solution of $Q$, with
$\vars = \V(Q)$, if all the relations in $Q$ are (trivially) valid
when replacing all variables $x$ in $\vars$ by their value $e(x)$. We
denote $\qsol{Q}$ the subset of $\Nat^{\V(Q)}$ consisting in all the
solutions of $Q$.

If $E \subseteq \Nat^{\vars}$ then $E \proj \uars$ is the projection
of $E$ over variables $\uars$, that is the subset of $\Nat^{\uars}$
obtained from $E$ by restricting the domain of its elements to
$\uars$. conversely, we use $E \uparrow \uars$ to denote the lifting
of $E$ to $\uars$, that is the largest subset $E'$ of $\Nat^{\uars}$
such that $E' \proj \vars = E$.

\begin{definition}[Net-abstraction]\label{def:abstraction} A triple $(N_1,Q,N_2)$
  is a {\em net-abstraction}, or simply an abstraction, if $N_1$,
  $N_2$ are nets with respective sets of places $P_1$, $P_2$ (we may
  have $P_1 \cap P_2 \not= \emptyset)$, $Q$ is a linear system of
  equations, and:
  \[
    \reach(N_1) = \left ( (\reach(N_2) \lift \vars) \cap (\qsol{Q}
      \lift \vars) \right ) \proj P_1 \qquad \text{where }
    \vars= \V(Q) \cup P_1 \cup P_2~.
  \]
\end{definition}

Intuitively, $N_2$ is an abstraction of $N_1$ (through $Q$) if, from
every reachable marking $m \in \reach(N_2)$, the markings obtained
from solutions of $Q$---restricted to those solutions such that
$x = m(x)$ for all ``place variable'' $x$ in $P_2$---are always
reachable in $N_1$. The definition also entails that all the markings
in $\reach(N_1)$ can be obtained this way.

\begin{theorem}[Net-abstractions from reductions]\label{the:soundness}
For any nets $N$, $N_1$, $N_2$:
\begin{enumerate}
% \item[] For any nets $N$, $N_1$, $N_2$:

\item $(N,\emptyset,N)$ is an abstraction;

\item If $(N_1,Q,N_2)$ is an abstraction then $(N_1,Q',N_3)$ is an abstraction if either:

 \begin{itemize}
 \item[(T)] $Q'= Q$ and $N_3$ is obtained from $N_2$ by removing a
   redundant transition (see Sect. \ref{sec:redTrans});

 \item[(R)] $Q'= (Q; k.p = l)$ and $N_3$ is
   obtained from $N_2$ by removing a redundant place $p$ and $k.p = l$
   is the associated marking equation (see Sect.~\ref{sec:redPlace});

 \item[(A)] $Q'= (Q; a = \Sigma_{p \in A} (p))$, where
   $a \not\in \V(Q)$ and $N_3$ is obtained from $N_2$ by agglomerating
   the places in $A$ as a new place, $a$ (see Sect.~\ref{sec:agglo});

 \item[(L)] $Q'= (Q; p \le k)$ and $N_3$ is obtained from $N_2$ by
   removal of a source-sink pair $(p,t)$ with $m_0(p) = k$ (see
   Sect.~\ref{sec:system}).
 \end{itemize}
\end{enumerate}
\end{theorem}

\begin{proof}
  Property (1) is obvious from Definition \ref{def:abstraction}.
  Property (2) is proved by case analysis. First, let $\vars = \V(Q) \cup P_1 \cup P_2$
  and $\uars = \vars \cup P_3$ and notice that for all candidate $(N_1,Q',N_3)$
  we have $\V(Q') \cup P_1 \cup P_3 = \uars$.
  Then, in each case, we know
  $(H): \reach(N_1) = (\reach(N_2) \lift \vars \cap \qsol{Q} \lift
  \vars) \proj P_1$ and we
  must prove
  $(G): \reach(N_1) = (\reach(N_3) \lift \uars \cap \qsol{Q'} \lift
  \uars) \proj P_1$.

  \vspace{3pt}
  \noindent\emph{Case (T)} : $Q'= Q$. By Th.~\ref{the:redTrans}, we
  have $P_3 = P_2$, hence $\vars =\uars$, and
  $\reach(N_3) = \reach(N_2)$.  Replacing $\reach(N_2)$ by
  $\reach(N_3)$ and $\vars$ by $\uars$ in (H) yields (G).

  \vspace{3pt}
  \noindent\emph{Case (R)} :
  By Th.~\ref{the:redPlace} we have :
  $\reach(N_2) = \reach(N_3) \lift P_2 \cap \qsol{k.p = l} \lift P_2$.
  replacing $\reach(N_2)$ by this value in $(H)$ yields
  $\reach(N_1) = ( (\reach(N_3) \lift P_2 \cap \qsol{k.p = l} \lift
  P_2) \lift \vars \cap \qsol{Q} \lift \vars) \proj P_1$. Since
  $P_2 \subseteq \vars$, we may safely lift to $\vars$ instead of
  $P_2$, so:
  $\reach(N_1) = ( \reach(N_3) \lift \vars \cap \qsol{k.p = l} \lift
  \vars \cap \qsol{Q} \lift \vars) \proj P_1$.  Which is equivalent
  to:
  $\reach(N_1) = (\reach(N_3) \lift \vars \cap \qsol{Q; k.p = l} \lift
  \vars) \proj P_1$, and equal to (G) since $P_3 \subseteq \vars$ and
  $Q'= (Q; k.p = l)$.

  \vspace{3pt}
  \noindent\emph{Case (A)}: Let $S_p$ denotes the value $\Sigma_{p \in A} (p)$.
  By Th.~\ref{the:agglo} we have:
  $\reach(N_2) = (\reach(N_3) \lift (P_2 \cup P_3) \cap \qsol{a = S_p}
  \lift (P_2 \cup P_3))\proj P_2$.  Replacing $\reach(N_2)$ by this
  value in $(H)$ yields:
  $\reach(N_1) = ( ((\reach(N_3) \lift (P_2 \cup P_3) \cap \qsol{a =
    S_p} \lift (P_2 \cup P_3)) \proj P_2) \lift \vars \cap \qsol{Q}
  \lift \vars) \proj P_1$.  Instead of $\vars$, we may lift to $\uars$
  since $\uars = \vars \cup \{a\}$, $a \not\in \V(Q)$ and
  $a \not\in P_1$, so:
  $\reach(N_1) = ( ((\reach(N_3) \lift (P_2 \cup P_3) \cap \qsol{a =
    S_p} \lift (P_2 \cup P_3)) \proj P_2) \lift \uars \cap \qsol{Q}
  \lift \uars) \proj P_1$.  Projection on $P_2$ may be omitted since
  $P_2 \cup P_3 = P_2 \cup \{a\}$ and $a \not\in \V(Q)$, leading to:

  $\reach(N_1) = ( (\reach(N_3) \lift (P_2 \cup P_3) \cap \qsol{a =
    S_p} \lift (P_2 \cup P_3)) \lift \uars \cap \qsol{Q}\lift \uars)
  \proj P_1$.

  \noindent Since $P_2 \cup P_3 \subseteq \uars$, this is equivalent
  to:
  $\reach(N_1) = (\reach(N_3) \lift \uars \cap \qsol{a = S_p} \lift
  \uars \cap \qsol{Q}\lift \uars) \proj P_1$. Grouping equations
  yields:
  $\reach(N_1) = ( \reach(N_3) \lift \uars \cap \qsol{Q;a = S_p} \lift
  \uars) \proj P_1$, which is equal to (G) since $Q'= (Q;a = S_p)$.

  \vspace{3pt}
  \noindent\emph{case (L)}:   The proof is similar to that of case (R)
  and is based on the relation
  $\reach(N_2) = \reach(N_3) \lift P_2 \cap \qsol{p \le k} \lift P_2$,
  obtained from Th.~\ref{the:sourcesink}. %\qed
\end{proof}

Theorem~\ref{the:soundness} states the correctness of our reduction
systems, since we can compose reductions sequentially and always
obtain a net-abstraction.
In particular, if a net $N$ is fully reducible, then we can derive a
system of linear equations $Q$ such that $(N, Q, \emptyset)$ is a
net-abstraction. In this case the reachable markings of $N$ are
exactly the solutions of $Q$, projected on the places of $N$. 
If the reduced net, say $N_r$, is not empty then each marking
$m \in \reach(N_r)$ represents a set of markings
$\qsol{Q}_m \subset \reach(N)$: the solution set of $Q$ in which the
places of the residual net are constrained as in $m$, and then
projected on the places of $N$. Moreover the family of sets
$\{\,\qsol{Q}_m~\mid~m \in \reach(N_r)\}$ is a partition of
$\reach(N)$.

%%%%%%%%%%%%%%%%%%%%%%%%%%%%%%%%%%%%%%%%%%%%%%%%%%%%%%%%%%%%%%%%
% \section{Counting markings}

\section{Counting Markings}
\label{sec:counting}

We consider the problem of counting the markings of a net $N$ from the
set of markings of the residual net $N_r$ and the (collected) system
of linear equations $Q$.
For totally reduced nets, counting the markings of $N$ resumes to that
of counting the number of solutions in non negative integer variables
of system $Q$.
For partially reduced nets, a similar process must be iterated over all markings $m$ reachable in $N_r$
(a better implementation will be discussed shortly).

\vspace{5pt}
\noindent\textbf{Available methods.}
Counting the number of integer solutions of a linear system of
equations (inequalities can always be represented by equations by the
addition of slack variables)
is an active area of research.

A method is proposed in~\cite{doi:10.1137/1.9781611972870.15},
implemented in the tool \azove, for the particular case where
variables take their values in $\{0,1\}$.
The method consists of building a Binary Decision Diagram for each
equation, using Shannon expansion, and then to compute their
conjunction (this is done with a specially tailored algorithm).
The number of paths
of the BDD gives the expected result.  Our experiments with \azove
show that, although the tool can be fast, its performances on larger
system heavily depend on the ordering chosen for the BDD variables, a
typical drawback of decision diagram based techniques. In any case,
its usage in our context would be limited to safe nets.

For the general case, the current state of the art can be found in the work of De~Loera
et al.~\cite{delorea,DELOERA20041273}
on counting lattice points in convex polytopes. Their approach is implemented in a tool
called \latte; it relies on algebraic and geometric methods;
namely the use of rational functions and the decomposition of cones into unimodular cones.
Our experiments with \latte show that it can be conveniently used
on systems with, say, less than $50$ variables. For instance, \latte is
strikingly fast (less than $1$s) at counting the number of solutions
of the system computed in Sect.~\ref{sec:impl-example-appl}.
Moreover, its running time does not generally depend on the constants
found in the system. As a consequence, computing the reachability
count for $10$ or, say, $10^{12}$ houses takes exactly the same time.

\vspace{5pt}
\noindent\textbf{An ad-hoc method.}
Though our experiments with \latte suffice to show that these
approaches are practicable, we implemented our own counting method.
Its main benefits over \latte, important for practical purposes, are
that it can handle systems with many variables (say thousands), though
it can be slower than \latte on small systems.
Another reason for finding an alternative to \latte is that it
provides no builtin support for parameterized systems, that is in the
situation where we need to count the solutions of many instances of
the same linear system differing only by some constants.

Our solution takes advantage of the stratified structure of the
systems obtained from reductions, and it relies on combinatorial rather
than geometric methods. While we cannot describe this tool in full
details, we illustrate our approach and the techniques involved on a
simple example.

Consider the system of equations obtained from the PN corresponding to
the dashed zone of Fig.~\ref{net:HouseConstruction-10}. This net
consists of places in the range $p_{18}$---$p_{27}$ and is reduced by
our system to a single place, $a_{13}$.  The subset of marking
equations related to this subnet is:
\begin{displaymath}
\begin{array}{l@{\hspace{6em}}l}
 R \vdash p_{19} = p_{20} &  R \vdash p_{27} = p_{23} + p_{26}\\
 A \vdash a_{4} = p_{21} + p_{18} &  A \vdash a_{11} = p_{23} + a_{5}\\
 A \vdash a_{5} = p_{22} + p_{20} &  A \vdash a_{13} = a_{11} + p_{26} \\
 R \vdash a_{4} = a_{5} + p_{23}\\
\end{array}
\hspace{4em} (Q)
\end{displaymath}
Assume place $a_{13}$ is marked with $n$ tokens.  Then, by
Th.~\ref{the:soundness}, the number of markings of the original net
corresponding with marking $a_{13} = n$ in the reduced net is the
number of non-negative integer solutions to system $(Q, a_{13} = n)$.
Let us define the function $A_{13} : \Nat \longrightarrow \Nat$
that computes that number.

We first simplify system $(Q)$. Note that no agglomeration is involved
in the redundancy (R) equations for $p_{19}$ and $p_{27}$, so these
equations have no effects on the marking count and can be omitted.
After elimination of variable $a_5$ and some rewriting, we obtain the
simplified system $(Q')$:
\begin{displaymath}
\begin{array}{l@{\hspace{6em}}l}
 A \vdash a_{4} = p_{21} + p_{18} & R \vdash a_{4} = a_{11} \\
  A \vdash a_{11} = p_{23} + p_{22} + p_{20} & A \vdash a_{13} = a_{11} + p_{26}
\end{array}
\hspace{4em} (Q')
\end{displaymath}
Let $\nos{k}{x}$ denote the expression ${x + k - 1} \choose {k - 1}$,
which denotes the number of ways to put $x$ tokens into $k$ slots.
The first equation is $a_{4} = p_{21} + p_{18}$. If $a_{4} = x$, its number of solutions
is $\nos{2}{x} = x + 1$.
The second equation is $a_{11} = p_{23} + p_{22} + p_{20}$. If $a_{11} = x$, its number of solutions
is $\nos{3}{x} = \frac{(x+2)(x+1)}{2}$.

Now consider the system consisting of the first two equations and the redundancy equation $a_4 = a_{11}$.
If $a_{11} = x$, its number of solutions is $\nos{2}{x} \times \nos{3}{x}$
(the variables in both equations being disjoint).
Finally, by noticing that $a_{11}$ can take any value between $0$ and $n$,
we get:
\begin{center}
\begin{math}
  \displaystyle
  A_{13}(n) = \sum_{a_{11}=0}^{n} \nos{2}{a_{11}} \times \nos{3}{a_{11}}
\end{math}
\end{center}

This expression is actually a polynomial, namely
\begin{math}
\frac{1}{8}n^4 + \frac{11}{12}n^3 + \frac{19}{8}n^2 + \frac{31}{12}n + 1
\end{math}.\\
By applying the same method to the whole net of
Fig.~\ref{net:HouseConstruction-10}, we obtain an expression involving
six summations, which can be reduced to the following 18th-degree polynomial in
the variable $X$ denoting the number of tokens in place $p_1$.

\[{\scriptscriptstyle\begin{split} 
      \tfrac{11}{19401132441600} X^{18} +
      \tfrac{1}{16582164480} X^{17} +
      \tfrac{2491}{836911595520} X^{16} +
      \tfrac{1409}{15567552000} X^{15}\\
      { } + 
      \tfrac{3972503}{2092278988800} X^{14} + 
      \tfrac{161351}{5535129600} X^{13} +
      \tfrac{32745953}{96566722560} X^{12} + 
      \tfrac{68229017}{22353408000} X^{11}\\
      { } +
      \tfrac{629730473}{29262643200} X^{10} + 
      \tfrac{83284643}{696729600} X^9 +
      \tfrac{3063053849}{5852528640} X^8 + 
      \tfrac{74566847}{41472000} X^7\\
      { } + 
      \tfrac{1505970381239}{313841848320} X^6 +
      \tfrac{32809178977}{3353011200} X^5 +
      \tfrac{259109541797}{17435658240} X^4 + 
      \tfrac{41924892461}{2594592000} X^3\\
      { } + 
      \tfrac{4496167537}{381180800} X^2 + 
      \tfrac{62925293}{12252240} X^1 + 
      1
\end{split}}
\]

In the general case of partially reduced nets, the computed polynomial
is a multivariate polynomial with at most as many variables as places
remaining in the residual net.  When that number of variables is too
large, the computation of the final polynomial is out of reach, and we
only make use of the intermediate algebraic term.

%%%%%%%%%%%%%%%%%%%%%%%%%%%%%%%%%%%%%%%%%%%%%%%%%%%%%%%%%%%%%%%%
% \section{Experimental Results}

\section{Computing Experiments}
\label{sec:experimental-results}

We integrated our reduction system and counting method with a state
space generation tool, {\em tedd}, in the framework of our \emph{TINA}
toolbox for analysis of Petri nets \cite{berthomieu2004tool}
(\url{www.laas.fr/tina}). Tool \emph{tedd} makes use of symbolic
exploration and stores markings in a Set Decision Diagram
\cite{Thierry-MiegPHK09}.  For counting markings in presence of
agglomerations, one has the choice between using the external tool
\latte or using our native counting method discussed in
Sect. \ref{sec:counting}.

\vspace{5pt}
\noindent\textbf{Benchmarks.}
Our benchmark is constituted of the full collection of Petri nets used
in the Model Checking Contest~\cite{mcc:2017,hk2017}. It includes $627$ nets,
organized into $82$ classes (simply called models). Each class includes several
nets (called instances) that typically differ by their initial
marking or by the number of components constituting the net.
The size of the nets vary widely, from $9$ to $50\,000$ places, $7$ to $200\,000$ transitions, and $20$
to $1\,000\,000$ arcs.  Most nets are ordinary (arcs have weight $1$)
but a significant number are generalized nets. Overall,
the collection provides a large number of PN with various structural
and behavioral characteristics, covering a large variety of use cases.

\vspace{5pt}
\noindent\textbf{Reduction ratio and prevalence.}
Our first results are about how well the reductions perform. We provide two
different reduction strategies: \verb+compact+, that applies all reductions
described in Sect.~\ref{sec:reductions}, and \verb+clean+, that only applies
removal of redundant places and transitions.
The reduction ratios on number of places (number of places before and after reduction)
for all the MCC instances are shown in Fig.~\ref{net:histo-ratio},
sorted in descending order. We overlay the results for our two reduction
strategies (the lower, in light color, for \verb+clean+ and the upper,
in dark, for \verb+compact+). We see that the impact of strategy \verb+clean+
alone is minor compared to \verb+compact+.
Globally, Fig.~\ref{net:histo-ratio} shows that reductions have a
significant impact on about half the models, with a very high impact
on about a quarter of them.  In particular, there is a surprisingly
high number of models that are totally reducible by our approach
(about $19$\% of the models are fully reducible).

\begin{figure}[!tb]
  {\centering
    \raisebox{-0.5\height}{\includegraphics[width=0.98\textwidth]{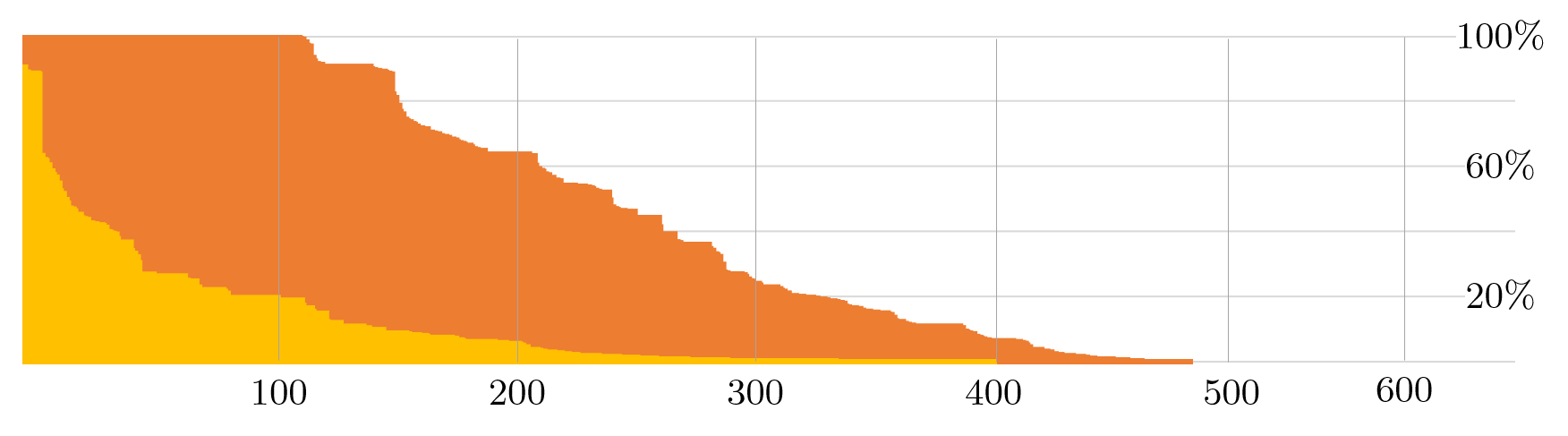}}
    \caption{Distribution of reduction ratios (place count) over the
      $627$ PN instances.\label{net:histo-ratio}}}
\end{figure}

\vspace{5pt}
\noindent\textbf{Computing time of reductions.}
Many of the reduction rules implemented have a cost polynomial in the size
of the net.
% For instance they require a simple graph traversal of the net structure.
The rule removing redundant places in the general case is more complex as
it requires to solve an integer programming problem.
For this reason  we limit its application to nets with less than $50$ places.
With this restriction, reductions are computed in a few seconds in most cases,
and in about $3$ minutes for the largest nets.
The restriction is necessary but, because of it, we do not reduce some
nets that would be fully reducible otherwise.

\vspace{5pt}
\noindent\textbf{Impact on the marking count problem.}
In our benchmark, there are $169$ models, out of $627$, for which 
no tool was ever able to compute a marking count.
With our method, we could count the markings of at least $14$ of them.

\begin{table}[!htb]
  \setlength\belowcaptionskip{-10pt}
  \setlength{\tabcolsep}{4pt}
  \begin{center}
    \begin{tabular}[c]{|l||r|l||r|r|r|r|}
      \hline
      Net instance & \multicolumn{2}{c||}{size}  & MCC   
      & \emph{tedd} &
                      \emph{tedd} & speed\\ \cline{2-3}
                   &  $\sharp$ places & $\sharp$ states  & (best) & native & \latte & up\\ %
      \hline
      \hline
      BART-050                & 11\,822 & 1.88e118       & 2\,800 & 346  & -  
                   & \rrq{2800}{346}\\
      BART-060                & 14\,132 & 8.50e141       & -    & 496  & -    
                   & $\infty$\\
      DLCround-13a            & 463   & 2.40e17        & 9    & 0.33 & -    
                   & \rrq{9}{0.33}\\
      FlexibleBarrier-22a     & 267   & 5.52e23        & 5    & 0.25 & -   
                   & \rrq{5}{0.25}\\
      NeighborGrid-d4n3m2c23  & 81    & 2.70e65        & 330  & 0.21 & 44   
                   & \rrq{330}{0.21}\\
      NeighborGrid-d5n4m1t35  & 1\,024  & 2.85e614       & -    & 340  & -    
                   & $\infty$\\
      Referendum-1000         & 3\,001  & 1.32e477       &  29  & 12   & -    
                   & \rrq{29}{12}\\
      RobotManipulation-00050 & 15    & 8.53e12        & 94   & 0.1  &
                                                                       0.17 
                   & \rrq{94}{0.1}\\
      RobotManipulation-10000 & 15    & 2.83e33        & -    & 102  &
                                                                       0.17 
                   & $\infty$ \\
      Diffusion2D-50N050      & 2\,500  & 4.22e105       & 1\,900 & 5.84 & -    
                   & \rrq{1900}{5.84}\\
      Diffusion2D-50N150      & 2\,500  & 2.67e36        & -    & 5.86 & -    
                   & $\infty$\\
      DLCshifumi-6a           & 3\,568  & 4.50e160       & 950  & 6.54 & -    
                   & \rrq{950}{6.54}\\
      Kanban-1000             & 16    & 1.42e30        & 240  & 0.11 &
                                                                       0.24 
                   & \rrq{240}{0.11}\\
      HouseConstruction-100   & 26    & 1.58e24        & 630  & 0.4  &
                                                                       0.85 
                   & \rrq{630}{0.4}\\
      HouseConstruction-500   & 26    & 2.67e36        & -    & 30   &
                                                                       0.85 
                   & $\infty$\\

        \hline
        \hline

        Airplane-4000           & 28\,019 & 2.18e12 & 2520 & 102   & - & \rrq{2520}{102}\\
        AutoFlight-48a          & 1127    & 1.61e51 &   19 &  3.57 & - & \rrq{19}{3.57}\\
        DES-60b                 &  519    & 8.35e22 & 2300 & 364   & - & \rrq{2300}{364}\\
        Peterson-4              & 480     & 6.30e8  &  470 &  35.5 & - & \rrq{470}{35.5}\\
        Peterson-5              & 834     & 1.37e11 &    - & 1200  & - & $\infty$\\
        \hline
    \end{tabular}
  \end{center}
    %% {\ }\\[0.5em]
  \caption{Computation times (in seconds) and speed-up for counting
    markings on some totally (top)
    and partially (bottom) reduced nets\label{tbl:reducible}}
\end{table}

If we concentrate on \emph{tractable nets}---instances managed by at least one tool in the
MCC 2017---our approach yields generally large improvements on the
time taken to count markings; sometimes orders of magnitude faster.
Table~\ref{tbl:reducible} (top) lists the CPU time (in seconds) for counting
the markings on a selection of fully reducible instances. We give the
best time obtained by a tool during the last MCC (third column) and
compare it with the time obtained with \emph{tedd}, using two
different ways of counting solutions (first with our own, native,
method then with \latte). We also give the resulting speed-up. These
times also include parsing and applying reductions. An absent
value ($-$) means that it cannot be computed in less than 1 hour with 16 Gb
of storage.

Concerning partially reducible nets, the improvements are less spectacular in general though
still significant. Counting markings in this case is more expensive than for totally reduced nets.
But, more importantly, we have to build in that case a representation of the state space of the residual net,
which is typically much more expensive than counting markings.
Furthermore, if using symbolic methods for that purpose, several other parameters come into
play that may impact the results, like the choice of an order on decision diagram variables or the particular
kind of diagrams used.
Nevertheless, improvements are clearly visible on a number of example models; 
some speedups are shown in Table \ref{tbl:reducible} (bottom).
Also, to minimize such side issues, instead of comparing {\em tedd} with \verb+compact+ reductions with the best tool
performing at the MCC, we compared it with {\em tedd} without reductions or with the weaker
\verb+clean+ strategy. In that case, \verb+compact+ reductions are almost always effective at reducing computing times.

Finally, there are also a few cases where applying reductions lower
performances, typically when the reduction ratio is very small.
For such quasi-irreducible nets, the time spent computing reductions
is obviously wasted.

%%%%%%%%%%%%%%%%%%%%%%%%%%%%%%%%%%%%%%%%%%%%%%%%%%%%%%%%%%%%%%%%
% \section{Conclusion}

\section{Related Work and Conclusion}

Our work relies on well understood structural reduction methods,
adapted here for the purpose of abstracting the state space of a
net. This is done by representing the effects of reductions by a
system of linear equations. To the best of our knowledge, reductions
have never been used for that purpose before.

Linear algebraic techniques are widely used in Petri net theory but,
again, not with our exact goals.
It is well known, for instance, that the state space of a net is
included in the solution set of its so-called ``state equation'', or
from a basis of marking invariants. But these solutions, though exact
in the special case of live marked graphs, yield approximations that
are too coarse.
Other works take advantage of marking invariants obtained from semiflows on places,
but typically for optimizing the representation of markings in explicit or symbolic
enumeration methods rather than for helping their enumeration,
see e.g. \cite{schmidt2003using,wolf2007generating}.
Finally, these methods are only remotely related to our.

Another set of related work concerns symbolic methods based on the use
of decision diagrams. Indeed they can be used to compute the state
space size. In such methods, markings are computed symbolically and
represented by the paths of some directed acyclic graph, which can be
counted efficiently.  Crucial for the applicability of these methods is
determining a ``good'' variable ordering for decision diagram
variables, one that maximizes sharing among the paths.  Unfortunately,
finding a convenient variable ordering may be an issue, and some
models are inherently without sharing. For example, the best symbolic
tools participating to the MCC can solve our illustrative example only
for $p_1 \le 100$, at a high cost, while we compute the result in a
fraction of a second for virtually any possible initial marking of
$p_1$.

Finally, though not aimed at counting markings nor relying on reductions,
the work reported in \cite{Stahl_AWPN} is certainly the closest to
our. It defines a method for decomposing the state space of a net into the
product of ``independent sets of submarkings''. The ideas discussed in the paper resemble
what we achieved with agglomeration. In fact, the running example
in \cite{Stahl_AWPN}, reproduced here in Figure \ref{fig:petri}, is a fully reducible net in our approach.
But no effective methods are proposed to compute decompositions.\\

\noindent\textbf{Concluding remarks.}
We propose a new symbolic approach for representing the state space of a PN relying
on systems of linear equations.
Our results show that the method is almost always effective at
reducing computing times and memory consumption for counting markings.
\iffalse
When added to symbolic enumeration tools like \emph{tedd}, reductions decreases both
the number of markings to be computed, what accelerates fixpoint computation,
and the number of places of the net, what may mitigate the effects of variable orders.
The case of fully reducible nets gives the most spectacular results, but
our approach is also useful in the general case of partially reducible nets.
\fi
Even more interesting is that our methods can be used together with traditional explicit
and symbolic enumeration methods, as well as with other abstraction techniques like symmetry
reductions for example. They can also help for other problems, like reachability analysis.

There are many opportunities for further research. For the close
future, we are investigating richer sets of reductions for counting
markings and application of the method to count not only the markings,
but also the number of transitions of the reachability
graph. Model-checking of linear reachability properties is another
obvious prospective application of our methods.
On the long term, a question to be investigated is how to obtain efficiently
fully equational descriptions of the state spaces of bounded Petri nets.

\newpage

\bibliography{counting}

\end{document}